\documentclass[runningheads]{llncs}

\usepackage[T1]{fontenc}

\usepackage{graphicx}
\usepackage[ruled,vlined,linesnumbered]{algorithm2e}
\usepackage{amsmath,amssymb,caption,subcaption, hyperref, pifont}

\begin{document}
\title{Gathering of asynchronous robots on circle with limited visibility using finite communication}
\titlerunning{Gathering robots on circle with limited visibility using finite communication}

\author{Avisek Sharma\inst{1}\orcidID{0000-0001-8940-392X} \and
Satakshi Ghosh\inst{2}\orcidID{0000-0003-1747-4037} \and
Buddhadeb Sau\inst{1}\orcidID{0000-0001-7008-6135}}
\authorrunning{A Sharma, S Ghosh, and Buddhadeb Sau}

\institute{Department of Mathematics, Jadavpur University, Kolkata, West Bengal, India 
\and 
Department of Basic Science and Humanities, International Institute of Information Technology, Bhubaneswar, Odisha, India\\
\email{aviseks.math.rs@jadavpuruniversity.in, satakshi.g@snuniv.ac.in, buddhadeb.sau@jadavpuruniversity.in}}
\maketitle              % typeset the header of the contribution
\begin{abstract}
This work deals with gathering problem for a set of autonomous, anonymous and homogeneous robots with limited visibility operating in on a continuous circle. The robots are initially placed at distinct positions forming a rotationally asymmetric configuration. The robots agree on the clockwise direction. In $\theta$-visibility model, a robot can only see those robots on the circle that are at an angular distance $<\theta$ from it. Di Luna \textit{et. al.} [DISC'20] have shown that, in $\pi/2$ visibility, gathering is impossible. Also, they provided an algorithm for robots with $\pi$ visibility, operating under a semi-synchronous scheduler. In the $\pi$ visibility model, only one point, the point at the angular distance $\pi$ is removed from the visibility. Ghosh \textit{et. al.} [SSS'23] provided a gathering algorithm for $\pi$ visibility model with robot having finite memory ($\mathcal{FSTA}$), operating under a special asynchronous scheduler.

If the robots can see all points on the circle, then the gathering can be done by electing a leader in weakest robot model under a fully asynchronous scheduler. But in the previous works, it was demonstrated that even removal of one point from the visibility makes gathering challenging. In both works, the robots had rigid movement. In this work, we propose an algorithm that solves the gathering problem under $\pi$-visibility model for robots that have finite communication ability ($\mathcal{FCOM}$). In this work the robot movement is non-rigid and the robots work under a fully asynchronous scheduler. 

\keywords{Limited visibility  \and Finite communication \and Asynchronous robots \and Continuous circle \and Distributed algorithms}
\end{abstract}

\section{Introduction}
Swarm robotics has been well-studied over the past two decades in distributed computing. Here a collection of simple and inexpensive robots (mobile computing units) collaboratively perform a task. In this line of research, we investigate to find the minimal capability required for the robots to complete a certain task. Generally robots are considered as dimensionless computational entities on the euclidean plane. They are expected to complete a given task by coordinating with each other. The robots operates through Look-Compute-Move (LCM) cycles. On activation a robot takes a snapshot of its surrounding in its vicinity (Look). Then taking information in the snapshot it runs an inbuilt algorithm and outputs a position to move (Compute). Next, it moves to that computed position (Move). The robots are generally autonomous (having no central control), anonymous (no unique identifier), identical (physically indistinguishable) and homogeneous (all  robots run the same distributed algorithm). Depending upon the time of activation there are different schedulers. In the fully synchronous setting (\textsc{FSync}), the time is divided into equal rounds and simultaneously executes the phases of their LCM cycles. The semi-synchronous model (\textsc{SSync}) is same as the \textsc{FSync} model, except for the fact that not all robots are necessarily activated in each round. The most general model is the fully asynchronous model (\textsc{ASync}) where there are no assumptions regarding the synchronization and duration of the robot actions.

There are different robots models depending on the capabilities of robots. The weakest robot model is $\mathcal{OBLOT}$. In this model, the robots are oblivious, i.e., robots have no persistent memory to remember its past actions or past configurations, and robots are silent, i.e., robots have no explicit communication ability. Next, robots may be equipped with a persistent light that can take colors from a predefined pallet consisting of a finite number of colors. If this light is visible to itself then it serves as a finite memory. This model is called $\mathcal{FSTA}$ (\cite{FLOCCHINI2016}). If this light is visible only to other robots, then this serves as a communication architecture. This model is called $\mathcal{FCOM}$ (\cite{OKUMURA2023114198}). If this light is visible to all robots including itself then the model is called $\mathcal{LUMI}$ (\cite{DAS2016}). In $\mathcal{LUMI}$ model, the robots have finite memory and also have communication ability.

In this work, we are concerned with the most fundamental coordination task, Gathering (\cite{AOSY1999,TERAI2023241}). Here, the robots are asked to meet at a point, not known beforehand. In general robots have full visibility, but in practice we may not have this luxury due to hardware limitations. In the limited visibility model (\cite{AOSY1999,FPSW2001,FLOCCHINI2005,PS2021}), a robot can see only those robots that are at a distance less than a constant $R>0$. The $R$ is called the visibility radius of the robots. In this work, the robots are operating in a continuous circle \cite{DILUNA2025}. Initially, the robots are placed at distinct positions of a circle. All robots agree on the clockwise direction. The initially forming configuration by the robots are rotationally asymmetric and has no multiplicity point, i.e., there is no location on the circle that is occupied by more than one robots. A configuration is rotationally asymmetric, means no non-trivial rotation with respect to the center will keep the configuration unchanged. In $\theta$-visibility model, the robots can see only those robots that have angular distance less than $\theta$ from it (See Figure~\ref{fig:vis}).

\begin{figure}[ht!]
    \centering
    \includegraphics[width=0.5\linewidth]{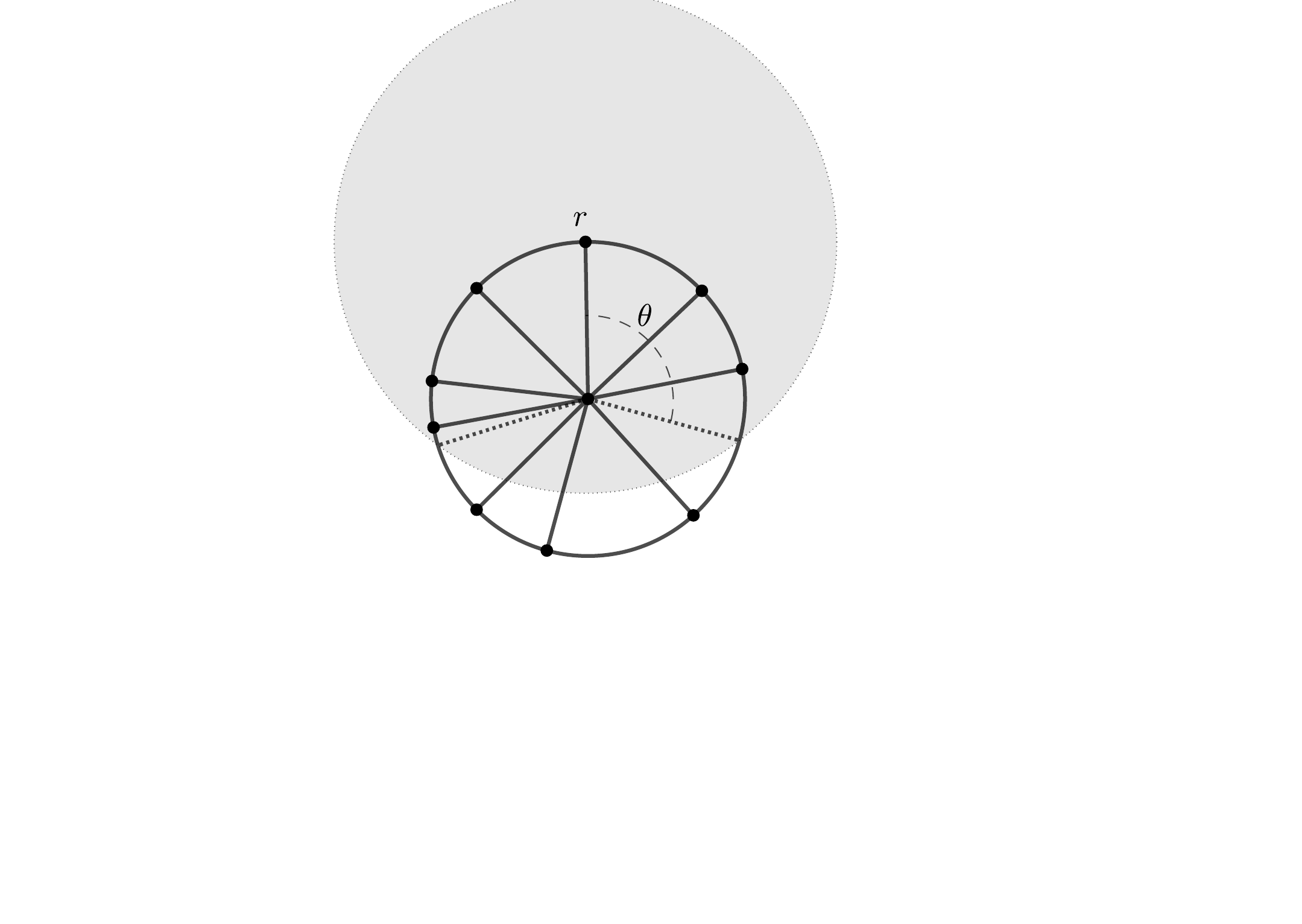}
    \caption{Visibility for robot $r$: black solid circles denote robots. Robot $r$ can see the robots that are at an angular distance $\theta$ from it.}
    \label{fig:vis}
\end{figure}

\section{Related work and our contribution}

Gathering problem by a group of robots has been studied extensively in different terrains and models (See the survey in \cite{Flocchini2019book}). Here we shall restrict our discussion to the continuous terrain when robots have limited visibility. In \cite{AOSY1999}, Ando \textit{et. al.} first time considered gathering in plane by robots with limited visibility. But there is an lighter version of gathering, called \textit{Point Convergence} which has been studied in the literature. The point convergence problem asks to design a distributed algorithm such that for any given $\epsilon>0$ there is a time $t_{\epsilon}$ such that for any time $t\ge t_{\epsilon}$ distance between any pair of robots is less than $\epsilon$. In works \cite{degener2011tight,DILUNA2025,FPSW2001,FLOCCHINI2005,GSGS2023,KB2011,kirkpatrick2024power,PS2021,souissi2009}, authors considered gathering under limited visibility and in the works \cite{AOSY1999,KB2011,kirkpatrick2024power}, authors considered point convergence with limited visibility. In \cite{DILUNA2025,GSGS2023} gathering on continuous circle has been considered by robots with limited visibility. In \cite{MONDAL2024}, Mondal \textit{et. al.} considered arbitrary pattern formation problem on continuous circle.

Next, we only focus into the details of the previous works where the gathering on continuous circle has been considered. In \cite{DILUNA2025}, Di Luna \textit{et. al.} considered first the gathering of robots on a continuous circle with limited visibility. They considered $\mathcal{OBLOT}$ robot model. They showed that rotationally asymmetric configuration at the initial is necessary. They first shown that if robots have $\pi/2$ visibility then the gathering is impossible. Then they proposed an algorithm to gather robots if robots have $\pi$ visibility. The $\pi$ visibility is a minimal limitation because it only removes one point from the visibility, the point on the circle at the angular distance $\pi$. We call the position at the angular distance $\pi$ from a robot as \textit{antipodal} position of it. While in full visibility, there is a gathering algorithm by electing a unique leader in a rotationally asymmetric configuration. Next, since robots are assumed to have multiplicity detection ability, so formation of one multiplicity is almost solves the problem. But removal of one point from the visibility, makes it challenging. First, election of one unique leader is not possible and also there is a chance of formation of two multiplicity points at antipodal positions. In \cite{DILUNA2025}, authors considered robots operating under semi-synchronous scheduler. Also, the they assumed the robots movements are rigid, i.e., when a robot decides to move at a point then adversary cannot stop its movement midway. They allow a robot to cross another robot while moving.

Next, in \cite{GSGS2023} Ghosh \textit{et. al.} attempted to solve it under the same model but in asynchronous scheduler. But authors could solve it for a special type of asynchronous scheduler. In addition, the authors traded synchrony with memory. They used the $\mathcal{FSTA}$ robot model. In both works, the rigidity of robots' movement remains the crucial and necessary assumption. So, in a practical scenario where maintaining rigidity is challenging, both of the previous solutions do not work. 

In this work, the authors have proposed a gathering algorithm under a fully asynchronous scheduler. Here, authors trade synchrony with communication. The robot model considered here is $\mathcal{FCOM}$. In addition, the authors removed the assumption of rigidity. The robots' movements are non-rigid. Since $\mathcal{LUMI}$ is an equally or more strong model than $\mathcal{FCOM}$, it also gives a solution for the $\mathcal{LUMI}$ model. Furthermore, to validate our contribution from the work in \cite{GSGS2023}, where the $\mathcal{FSTA}$ model has been used, in Section~\ref{sec:comp}, we have shown that for non-rigid movements $\mathcal{FCOM}$ is not a stronger model than $\mathcal{FSTA}$\footnote{See \cite{flocchiniOPODIS2023} for better understanding of model comparisons.}. We have defined a problem in the same terrain (continuous circle) and same visibility model that is solvable in model $\mathcal{FSTA}$ but not solvable in $\mathcal{FCOM}$.

\paragraph{\bf Our Contribution:} 
In this work, a set of robots are initially deployed at distinct locations on a continuous circle, initially forming a rotationally asymmetric configuration. The robots can freely move on the perimeter of the circle. The robots cannot see all points of the circle. The robots are operating under the $\pi$ visibility model, i.e., each robot cannot see the point at an angular distance $\pi$ from it. The robots agree on the clockwise direction. Robots have a weak multiplicity detection ability, i.e., from a snapshot, they can detect a multiplicity point but cannot count the number robots occupying that point. The robot model is $\mathcal{FCOM}$, i.e., each robot is equipped with a persistent external light that can take colors from a predefined palette consisting of finitely many colors. The robot's movements are non-rigid, i.e., a robot's movement can be stopped by an adversary midway. Although, for the sake of fairness, every time a robot attempts to move, it must move a constant distance $\delta>0$ towards its destination location. The robots are operating under a fully asynchronous scheduler. In addition, to justify our contribution, we prove that model $\mathcal{FCOM}$ is not as strong as $\mathcal{FSTA}$ (a model used in previous work) when robot movements are non-rigid. This work first proposes a gathering algorithm for the $\pi$ visibility model under a fully asynchronous scheduler and with non-rigid robot movements. The Table~\ref{tab:comp} compares our contribution with the other relevant works.

\begin{table}[ht!]
    \centering
    \footnotesize
    \caption{Comparison table}
    \label{tab:comp}
    \begin{tabular}{|c|p{1.5cm}|p{3cm}|c|}
    \hline
      Work   & Model & Scheduler  & Rigidity\\
    \hline
       Di Luna \textit{et. al.} \cite{DILUNA2025} & $\mathcal{OBLOT}$  & Semi-synchronous & Yes \\
        \hline
        Ghosh \textit{et. al.} \cite{GSGS2023} & $\mathcal{FSTA}$  & A special asynchronous scheduler & Yes\\
        \hline
        This work & $\mathcal{FCOM}$ & Fully asynchronous Scheduler & No\\ 
        \hline
    \end{tabular}
    
\end{table}

\section{Model definition and preliminaries}
Let a set of $k$ robots be initially placed on a circle $C$ on euclidean plane. The robots looks identical from outside. The robots are autonomous and homogeneous. All robots are governed by the same algorithm. The robots agree on the clockwise direction. The robots can freely move on the circle. To widen the scope of the applicability, we assume that the robots cannot cross each other on the circle. If a robot meets another robot in the middle of its move, its movement is paused until it is activated again. The robots movements are not rigid. That is, an adversary can stop a robot's movement at any point. But for the sake of fairness, we assume that once a robot starts moving it will at least make a constant $\delta>0$ amount movement towards its computed destination. Each robot is equipped with an external light that can take finitely many predefined colors from the palate $P$. The colors of the light is persistent, that i.e., it remains unchanged unless it is changed by the robot according to the algorithm.

We assume that initially each robot is occupying a distinct position on $C$. Let $a$ and $b$ be two distinct points on $C$. The clockwise angular distance between $a$ and $b$, denoted as $\alpha(a,b)$ is the angle subtended at the center of $C$ by while traversing from $a$ to $b$ on $C$ in clockwise direction. Similarly, $\overline{\alpha}(a,b)$ ($=2\pi-\alpha(a,b)$) denotes the counterclockwise angular distance between $a$ and $b$. For two robots $r_1$ and $r_2$ located on $C$ at points $a$ and $b$ respectively, we denote the clockwise (resp, anticlockwise) angular distance between them as $\alpha(r_1,r_2)$ (resp, $\overline{\alpha}(r_1,r_2)$) which is equal to $\alpha(a,b)$ (resp, $\overline{\alpha}(a,b)$). 

We assume the robots have limited visibility. Each robot can see all robots at an angular distance $<\pi$. The point at the angular distance $\pi$ is not visible to a robot. The robots operate through Look-Compute-Move cycles. Upon activation in the look phase, the robots scan the visible part of the circle and obtain the position and color of the light of the visibility. In compute phase, the robots run the inbuilt algorithm and obtain a position and a color $c$ from the palate $P$. In the move phase first it sets the color of its light at $c$ and then move the the computed position.

The robots are operating under a fully asynchronous scheduler. The robots activate independently and execute their phases of LCM cycle independently. The duration of the phases can take unbounded amount of time. For a robot the duration of the phases can be different from other robots and also from its own earlier cycle. 

A finite set ordered pairs $(p,\omega(p))$ is said to be a \textit{configuration of robots} (or, simply a \textit{configuration}) where the first quadrant of each pair $p$ is a unique point on $C$ occupied by a robot and $\omega(p)=\{1,mult\}$. The $\omega(p)=1$ denotes that $p$ is occupied by only one robot and $\omega(p)=mult$ denotes that $p$ is occupied by more than one robot. A point $p$ having $\omega(p)=mult$ is said to be a \textit{multiplicity} point. If a configuration has no multiplicity point, then we can ignore the second quadrant of the members of $\mathcal{C}$. For such a case, a configuration is equivalent to a finite set of points on $C$.

In a snapshot of a robot $r$, it receives ordered pairs $(p,\omega(p))$, where $p$ is a point visible to $r$ occupied by a robot. Thus, robots have weak multiplicity detection ability, i.e., a robot can detect whether a point $p$ is occupied by one robot or more than one robot from the value of $\omega(p)$, but it cannot count the number of robots occupying a multiplicity point.

Let $\mathcal C$ be a configuration with no multiplicity point. Let $S$ be a finite set of points on $C$ occupied by the robots in $\mathcal C$. The $\mathcal C$ is said to be \textit{rotational symmetric} if there is a nontrivial rotation of points of $S$ with respect to the center which leaves the configuration unchanged. A configuration with no multiplicity point is said to be \textit{rotationally asymmetric} if it is not rotationally symmetric. A robot $r$ is said to be an antipodal robot if there exists a robot $r'$ on the angular distance $\pi$ of the robot. In such a case, $r$ and $r'$ are said to be \textit{antipodal} robots to each other.

Let $r$ be a robot in a given configuration with no multiplicity point and let $r_1, r_2,\dots, r_n$ be the other robots on the circle in clockwise order. Then the angular sequence for robot $r$ is the sequence $(\alpha(r,r_1),\alpha(r_1,r_2),\alpha(r_2,r_3),\dots,\alpha(r_n,r))$. We denote this sequence as $\mathcal{S}(r)$. Further, $\mathcal{S}(r,r_i)$ denote the following subsequence of $\mathcal{S}(r)$: $(\alpha(r,r_1),\alpha(r_1,r_2),\alpha(r_2,r_3),\dots,\alpha(r_{i-1},r_i))$. Further, we call $\alpha(r,r_1)$ as the leading angle of $r$. We denote the leading angle of a robot $r$ as $\lambda(r)$.

Note that as the configuration is initially rotationally asymmetric, by results from the paper \cite{DILUNA2025} we can say that all the robots have distinct angle sequences.
\begin{definition}[Lexicographic Ordering]
Let $\tilde{\alpha} = (\alpha_1,\dots,\alpha_n)$ and $\tilde{\beta}=$ $(\beta_1,\dots,\beta_n)$ be two finite sequences of reals of same length. Then $\tilde{\alpha}$ is said to be lexicographically strictly smaller sequence than $\tilde{\beta}$ if $\alpha_1<\beta_1$ or there exists $1<k<n$ such that $\alpha_i=\beta_i$ for all $i=1,2,\dots,k$ and $\alpha_{k+1}<\beta_{k+1}$. $\tilde{\alpha}$ is said to be lexicographically smaller sequence than $\tilde{\beta}$ if either $\tilde{\alpha}=\tilde{\beta}$ or $\tilde{\alpha}$ is lexicographically strictly smaller sequence than $\tilde{\beta}$.
\end{definition}
\begin{definition}[True leader]
In a configuration with no multiplicity point, a robot with lexicographically smallest angular sequence is called a true leader.
\end{definition}
If the configuration is rotationally asymmetric and contains no multiplicity point, there exists exactly one robot which has strictly the smallest lexicographic angle sequence. Hence there is only one true leader for such a configuration.
Since a robot on the circle cannot see whether its antipodal position is occupied by a robot or not. So a robot can assume two things: 1)~the antipodal position is empty, let's call this configuration $\mathcal C_0(r)$ 2)~the antipodal position is nonempty, let's call this configuration $\mathcal C_1(r)$. So a robot $r$ can form two angular sequences. One considering $C_0(r)$ configuration and another considering $C_1(r)$. The next two definitions are from the viewpoint of a robot. If the true leader robot can confirm itself as the true leader, we call it a \textit{cognizant} leader. If the true leader or some other robot has an ambiguity of being a true leader depending on the possibility of $\mathcal C_0$ or $\mathcal C_1$ configuration, then we call it \textit{undecided} leader. For a robot $r$, there may be the following possibilities. 
\begin{itemize}
    \item Possibility-1: $\mathcal C_0(r)$ configuration has rotational symmetry, so $\mathcal C_1(r)$ is the only possible configuration.
    \item Possibility-2: $\mathcal C_1(r)$ configuration has rotational symmetry, so $\mathcal C_0(r)$ is the only possible configuration.
    \item Possibility-3: Both $\mathcal C_0(r)$ and $\mathcal C_1(r)$ has no rotational symmetry, so both $\mathcal C_0(r)$ and $\mathcal C_1(r)$ can be possible configurations.
\end{itemize}
\begin{definition}[Cognizant leader]
A robot $r$ in a rotationally asymmetric configuration with no multiplicity point is called a cognizant leader if $r$ is the true leader in any possible configurations.
\end{definition}

Note that, the cognizant leader is definitely the true leader of the configuration. Thus, if the configuration is asymmetric and contains no multiplicity point, there is at most one cognizant leader.

\begin{definition}[Undecided leader]
A robot $r$ in a rotationally asymmetric configuration with no multiplicity point is called an undecided leader if both $C_0(r)$ and $C_1(r)$ are possible configurations and $r$ is a true leader in one configuration but not in another.
\end{definition}
\begin{definition}[Follower robot]
A robot in an asymmetric configuration with no multiplicity point is said to be a follower robot if it is neither a cognizant leader nor an undecided leader.
\end{definition}
\begin{definition}[Expected leader]
A robot in an asymmetric configuration with no multiplicity point is said to be an expected leader if it is not a follower robot. That is, an expected leader is either a cognizant leader or an undecided leader.
\end{definition}
Note that the above definitions are set in such a way that the cognizant leader (or, an undecided leader or, a follower robot) can recognize itself as the cognizant leader (or, an undecided leader or, a follower robot).  
\begin{definition}
For two robots $r$ and $r'$ situated at different positions on the circle with $\alpha(r,r')=\theta$, we define $[r,r']$ as the set of points $x$ on the circle such that  $0\le \alpha(r,x) \le \theta$ and $(r,r')$ as the set of points $x$ on the circle such that  $0 < \alpha(r,x) < \theta$.
\end{definition}
\begin{definition}
Let $r$ be a robot in a configuration, then a robot $r_1$ is said to be situated at the left of $r$ if $\alpha(r,r_1)>\pi$ and said to be at right if $\alpha(r,r_1)<\pi$.
\end{definition}
Note that, the true leader of the configuration can become an undecided leader and a robot other than the true leader can also become an undecided leader. The next results lead us to find the maximum possible number of expected leaders in a rotationally asymmetric configuration with no multiplicity points.
\begin{proposition}\label{lemma1}
Let $(\alpha_1,\dots,\alpha_k)$ be the angular sequence of the true leader, say $r_0$, in a rotationally asymmetric configuration with no multiplicity point. Then there cannot be another robot $r'$ with the following properties.
\begin{enumerate}
    \item $r'$ is at left side to $r_0$,
    \item $\mathcal{S}(r',r_0)=(\alpha_1,\alpha_2,\dots,\alpha_i)$.
  
\end{enumerate}
\end{proposition}
\begin{proof}
 %\paragraph*{Proof of Proposition~\ref{lemma1}}
We prove this result by contradiction. If possible let there be such a robot $r'$. Then note that $\mathcal{S}(r_0)$ is $\tilde{\alpha}=(\alpha_1,\dots,\alpha_i,\alpha_{i+1},\dots,\alpha_{k-i},\alpha_1,\dots,\alpha_i)$ and $\mathcal{S}(r')$ is $\tilde{\alpha}_1=(\alpha_1,\dots,\alpha_i,\alpha_1,\dots,\alpha_i,\alpha_{i+1},\dots,\alpha_{k-i})$. Since $r_0$ has the strictly smallest angular sequence, so $\alpha_1$ is the smallest angle in the configuration and also $\alpha_{i+1}\le \alpha_1$, which leads to $\alpha_{i+1}=\alpha_1$. Next, we show that $\alpha_2=\alpha_{i+2}$. Since the $\tilde{\alpha}$ is strictly the smallest angular sequence so $\alpha_{i+2}\le \alpha_2$. If $\alpha_{i+2}<\alpha_2$, then the angular sequence $(\alpha_{i+1},\dots,\alpha_{k-i},\alpha_1,\dots,\alpha_i,\alpha_1,\dots,\alpha_i)$ is smaller than $\tilde{\alpha}$, which is a contradiction. Hence $\alpha_2=\alpha_{i+2}$. Therefore by a similar argument, we can show that $\alpha_{i+j}=\alpha_j$ for $j=3,\dots,i$. Now if $2i= k$ then we see that $\tilde{\alpha}=\tilde{\alpha}_1$, which contradicts the fact that the configuration is rotationally asymmetric. Otherwise, proceeding similarly we can show that $\alpha_{2i+j}=\alpha_j$, for $j=1,2,\dots,i$. Repeating the same argument we can show that $\alpha_{pi+j}=\alpha_j$, for $j=1,2,\dots,i$. Since there are finitely many angles, so after finite number of steps we must end up having that $k=ti$ where $t\ge 2$ and $\tilde{\alpha}=\tilde{\alpha}_1=(\alpha_1,\dots,\alpha_i,\alpha_1,\dots,\alpha_i,\dots,\alpha_1,\dots,\alpha_i)$, which is again a contradiction. \qed
 \end{proof}

Next, we state a simple observation in the following Proposition~\ref{lemma00}.
\begin{proposition}\label{lemma00}
Suppose there is a rotationally asymmetric configuration with no multiplicity point with the true leader, $r_0$ (say), then on including a robot, say $r$, on the circle at an empty point without bringing any rotational symmetry, the true leader of the new configuration must be in $[r_0,r]$.    
\end{proposition}

For an undecided leader $r$, there may be two possibilities. The first one is when $r$ is a true leader in $\mathcal{C}_0(r)$ configuration but not in $\mathcal{C}_1(r)$. The second one is when $r$ is a true leader in $\mathcal{C}_1(r)$ configuration but not in $\mathcal{C}_0(r)$. We show that the second possibility can not occur. We formally state the result in the following Proposition.

\begin{proposition}\label{lemma30}
If a robot $r$ is a undecided leader in asymmetric configuration with no multiplicity point, $r$ is the true leader in $\mathcal{C}_0(r)$ configuration but $r$ is not the true leader in $\mathcal{C}_1(r)$ configuration.
\end{proposition}
From Proposition~\ref{lemma30} one can observe that if the true leader of an asymmetric configuration with no multiplicity point is an undecided leader then its antipodal position must be empty. And also if a robot, which is not the true leader of the configuration, becomes an undecided leader then its antipodal position must be non-empty. We record these observations in the following Corollaries.

\begin{corollary}\label{Cor1}
In a rotationally asymmetric configuration with no multiplicity point if the true leader of the configuration is an undecided leader then its antipodal position must be empty.
\end{corollary}

\begin{corollary}\label{Cor2}
In a rotationally asymmetric configuration with no multiplicity point if a robot other than the true leader becomes an undecided leader then its antipodal position must be non-empty.
\end{corollary}

In Proposition~\ref{lemma44}, we record an important property regarding the locations of the expected leaders, and proposition~\ref{lemma3} determines the total number of possible expected leaders in a rotationally asymmetric configuration without multiplicity points.

\begin{proposition}\label{lemma44}
Let $\mathcal{C}$ be a rotationally asymmetric configuration with no multiplicity point and $L$ be the true leader of the configuration, then another expected leader $r$ of $\mathcal{C}$ must satisfy $\alpha(L,r)\ge\pi$.
\end{proposition}
\begin{proof}
 %\paragraph*{Proof of Proposition~\ref{lemma44}}
If possible let there be another expected leader $r$ such that $\alpha(L,r)<\pi$. Since $r$ is not the leader of the configuration, so $r$ must be an undecided leader. Hence from Proposition~\ref{lemma30}
we have that $r$ is true leader in $\mathcal C_0(r)$ configuration and $L$ is the true leader of $\mathcal C_1(r)$ configuration. This contradicts the Proposition~\ref{lemma00}. \qed

 \end{proof}

\begin{proposition}\label{lemma3}
For any given rotationally asymmetric configuration with no multiplicity point, there can be at most one undecided leader other than the true leader.
\end{proposition}
\begin{proof}
 % \paragraph*{Proof of Proposition~\ref{lemma3}}
Let $\mathcal{C}$ be the given configuration. If possible let $r_1$ and $r_2$ be two undecided leaders other than the true leader, say $L$. From Proposition~\ref{lemma44}, First we have that $\alpha(L,r_i)\ge\pi$, for $i=1,2$. Without loss of generality we assume $\alpha(L,r_1)<\alpha(L,r_2)$. Since $r_i$s are undecided leaders so from Proposition~\ref{lemma30}, their antipodal position is non-empty. Let for each $i$, $r_i'$ be the antipodal robot of $r_i$. There are two exhaustive cases. First one is $\alpha(L,r_1)=\pi$ and second one is $\alpha(L,r_1)>\pi$.

\textbf{Case-I}: $\alpha(L,r_1)=\pi$

\textit{Case-IA:} In this case let the first angle of the angle sequences of $r_1$ and $L$ be different. Let the first angle in $\mathcal{S}(L)$ and $\mathcal{S}(r_1)$ are $\theta$ and $\theta_1$ respectively, then $\theta<\theta_1$. Since $r_2'$ cannot be the clockwise neighbor of $L$, so $r_2$ can see the leading angle of $L$. Therefore in order to become an undecided leader, the $\lambda(r_2)$ has to be $\theta$. Now there may be two cases. Firstly either the clockwise neighbor of $r_2$ is $L$ or not. If the clockwise neighbor of $r_2$ is $L$ then from Proposition~\ref{lemma00}, $L$ does not remain the true leader of the configuration, which is a contradiction. Secondly, if the clockwise neighbor of $r_2$ is not $L$ then $r_1$ can see the leading angle of $r_2$, that is $\theta$ which is smaller than $\lambda(r_1)$. This gives $\mathcal{C}\smallsetminus L$ cannot have $r_1$ as a true leader. This contradicts the fact that $r_1$ is an undecided leader.

\textit{Case-IB:} Let the first $t$ angles of $\mathcal{S}(L)$ and $\mathcal{S}(r_1)$ are same and $(t+1)^{th}$ angles are different. Let first $t$ clockwise neighbors of $L$ in clockwise order are $x_1,x_2,\dots,x_t$ and first $t$ clockwise neighbors of $r_1$ in clockwise order are $x_1',x_2',\dots,x_t'$. First, we show that $r_2$ is none of $x_i'$s. If possible let $r_2=x_i'$ where $i<t$. Then $r_2$ can see the first $i-1$ angles of $\mathcal{S}(L)$ and $\mathcal{S}(r_1)$. Since $i^{th}$ clockwise neighbor of $r_2$ is at the left of $r_2'$, $r_2$ can see its first $i$ angles of $\mathcal{S}(r_2)$ in original configuration. Now first we observe that the first $i-1$ angles of $r_2$ and $L$ are the same. If not then either $r_2$ would not be an undecided leader or $L$ would not be leader of the configuration. Now we see that $i^{th}$ angle of $\mathcal{S}(L)$ and $\mathcal{S}(r_2)$ is also same. If possible let $i^{th}$ angle of $\mathcal{S}(L)$ and $\mathcal{S}(r_2)$ are $\theta_i$ and $\alpha_i$ respectively, and $\alpha_i>\theta_i$ (Note that the case $\alpha_i<\theta_i$ gets excluded from the fact that $L$ is the true leader of the configuration). Since $i^{th}$ angle of $r_1$ is also $\theta_i$ which is visible by $r_2$, so then $r_2$ would not be undecided leader. Hence first $i$ angles of $\mathcal{S}(r_1)$ and $\mathcal{S}(r_2)$ are same. Therefore from Proposition~\ref{lemma00}, the $r_2$ cannot be an undecided leader. Hence for each $i<t$, $r_2\ne x_i'$. Now it is easy to observe that $r_2\ne x_t$, because otherwise, $\lambda(r_2)>\lambda(L)$. This implies that $r_2$ is not an undecided leader.

Hence for each $1\le i\le t$, $r_2\ne x_i'$. Thus note that $r_2$ can see the first $t+1$ angles of $\mathcal{S}(L)$ and in order to remain an undecided leader first $t+1$ angles of $\mathcal{S}(r_2)$ should match with it. Now there are two cases, either $(t+1)^{th}$ clockwise neighbor of $r_2$ is $L$ or not. If not then $r_1$ can see the first $t+1$ angles of $\mathcal{S}(r_2)$. And $(t+1)^{th}$ angle of $\mathcal{S}(r_2)$ is smaller than same of $\mathcal{S}(r_1)$. Therefore $r_1$ would not be an undecided leader. In other case if $(t+1)^{th}$ clockwise neighbor of $r_2$ is $L$ then from Proposition~\ref{lemma00}, $L$ does not remain leader of the configuration. Hence for case-I, we end up having a contradiction if there is more than one undecided leader other than a true leader.

\textbf{Case-II}: $\alpha(L,r_1)>\pi$

\textit{Case-IIA:} In this case let the first angle of the angle sequences of $r_1$ and $L$ be different. Let the first angle in $\mathcal{S}(L)$ and $\mathcal{S}(r_1)$ are $\theta$ and $\theta_1$ respectively, then $\theta<\theta_1$. Since in this case $r_2'$ cannot be the clockwise neighbor of $L$, so $r_2$ can see the leading angle of $L$. Therefore in order to become an undecided leader, the $\lambda(r_2)$ has to be $\theta$. Now $r_1$ can see the leading angle of $r_2$, that is $\theta$ which is $<\lambda(r_1)$. This gives $\mathcal{C}\smallsetminus r_1'$ cannot have $r_1$ as true leader. This contradicts the fact that $r_1$ is undecided leader.

\textit{Case-IIB:} Let the first $t$ angles of $\mathcal{S}(L)$ and $\mathcal{S}(r_1)$ are same and $(t+1)^{th}$ angles are different. Let first $t$ clockwise neighbors of $L$ in clockwise order are $x_1,x_2,\dots,x_t$ and $t$ clockwise neighbors of $r_1$ in clockwise order are $x_1',x_2',\dots,x_t'$. Since $r_1$ is not antipodal of $L$, so $r_2'$ cannot be the clockwise first neighbor of $L$. Borrowing the argument from Case-I we can conclude that $r_2$ is none of $x_i'$s. In other cases, $r_2$ can see the first $t+1$ angles of $\mathcal{S}(L)$ and in order to be an undecided leader first $t+1$ angles of $\mathcal{S}(r_2)$ should coincide with the same with $\mathcal{S}(L)$. Now we see that $r_1$ can see the first $t+1$ angles of $\mathcal{S}(r_2)$. If not then, since $r_1'$ is at right to $L$, so $L$ will be at most $t^{th}$ neighbor of $r_1$. This implies $L$ is not the true leader of the configuration from Proposition~\ref{lemma00}. So $r_1$ can see the first $t+1$ angles of $\mathcal{S}(r_2)$. Since $(t+1)^{th}$ angle of $\mathcal{S}(r_2)$ is smaller than the same of $\mathcal{S}(r_1)$ then $r_1$ does not remain an undecided leader of the configuration. Which is a contradiction. 

Hence it is proved that there cannot be more than one undecided leader other than the true leader of the configuration. \qed

 \end{proof}

Let $\mathcal{C}$ be a rotationally asymmetric configuration with no multiplicity point. Then $\mathcal{C}$ will have a true leader and from above Proposition~\ref{lemma3}, there can be at most one more undecided leader. Hence we can have the following four exhaustive cases for $\mathcal{C}$.
\begin{enumerate}
    \item $\mathcal{C}$ has exactly one expected leader and that is a cognizant leader.
    \item $\mathcal{C}$ has exactly one expected leader and that is a undecided leader.
    \item $\mathcal{C}$ has exactly two expected leaders and both are undecided leaders.
    \item $\mathcal{C}$ has exactly two expected leaders. One of them is a cognizant leader and another one is a undecided leader.
\end{enumerate}
The Fig.~\ref{gather0} gives the existence of all four above cases. For the third case, we observe two properties in the following Propositions.

\begin{figure}[ht!]
    \centering
    \begin{subfigure}[t]{0.45\textwidth}
        \centering
        \includegraphics[height=1.2in]{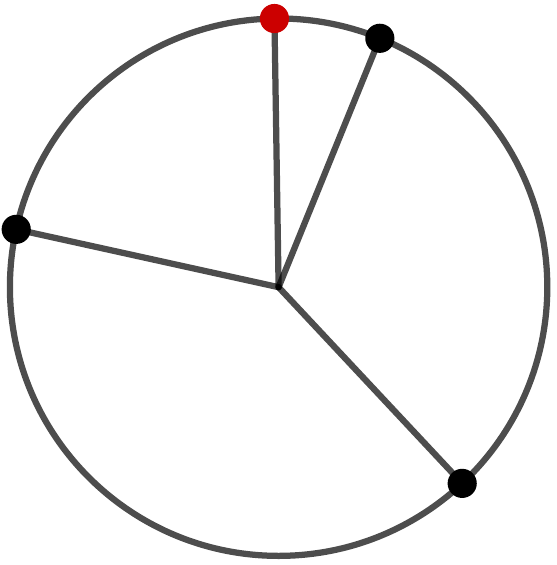}
        \caption{One cognizant leader}
    \end{subfigure}
    \begin{subfigure}[t]{0.45\textwidth}
        \centering
        \includegraphics[height=1.2in]{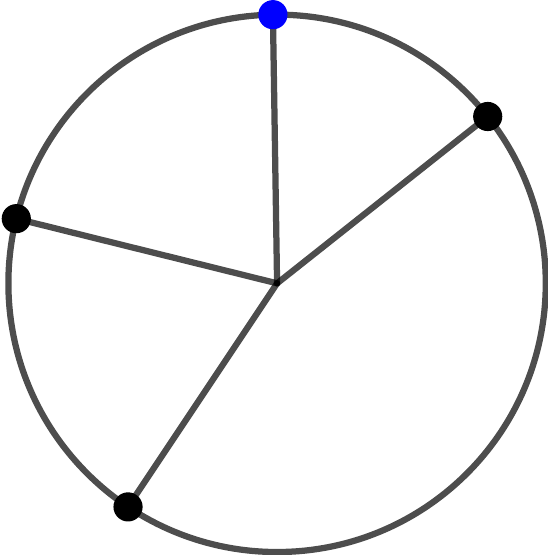}
        \caption{One undecided leader}
    \end{subfigure}
    \begin{subfigure}[t]{0.45\textwidth}
        \centering
        \includegraphics[height=1.2in]{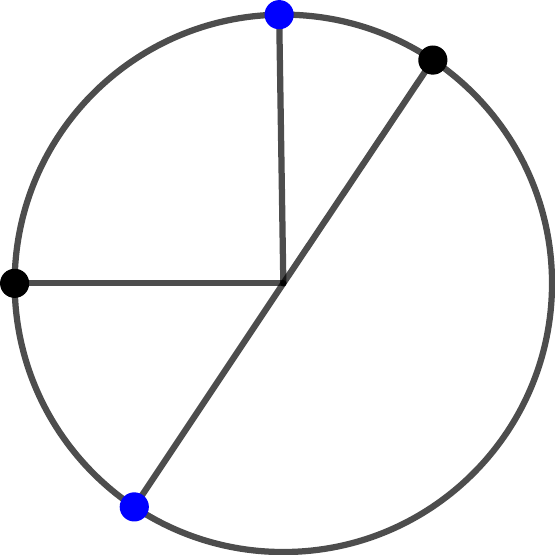}
        \caption{Two undecided leaders}
    \end{subfigure}
    \begin{subfigure}[t]{0.45\textwidth}
        \centering
        \includegraphics[height=1.2in]{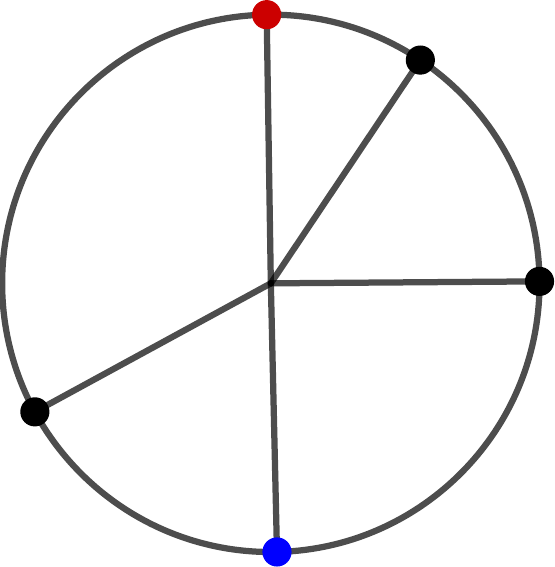}
        \caption{One cognizant and one undecided leader}
    \end{subfigure}
    \caption{Different possibilities of expected leaders; Red colored and blue colored solid circles respectively denote a cognizant leader and an undecided leader.}
    \label{gather0}
\end{figure}

Next, we record more results in Proposition~\ref{lemma2} and Proposition~\ref{lemma3oo} regarding the positions of the expected leaders, which guides the design of the proposed algorithm.

\begin{proposition}\label{lemma2}
If for a rotationally asymmetric configuration with no multiplicity point if there are two undecided leaders then they can not be antipodal of each other.
\end{proposition}
\begin{proof}
 %\paragraph*{Proof of Proposition~\ref{lemma2}}
Let $p$ and $q$ be two undecided leaders of a rotational asymmetric configuration with no multiplicity point. Now one of $p$ and $q$ must be the true leader of the configuration. Without loss of generality let $p$ be the leader of the configuration. Then from the Corollary~\ref{Cor1}, antipodal position of $p$ must be empty. Hence another undecided leader $q$ cannot be at the antipodal position of $p$. \qed
 \end{proof}

 \begin{proposition}\label{lemma3oo}
If for a rotationally asymmetric configuration with no multiplicity point if there are two undecided leaders then their clockwise neighbors cannot be antipodal to each other.
 \end{proposition}
 \begin{proof}
 %\paragraph*{Proof of proposition~\ref{lemma3oo}}
If possible let the clockwise neighbors of two undecided leaders be antipodal to each other. Let $L$ be the true leader of the configuration and $r$ be another undecided leader. Then from Proposition~\ref{lemma44} and Proposition~\ref{lemma2} we have $\alpha(L,r)>\pi$. Then $\lambda(L)<\lambda(r)$ (Figure~\ref{Fig:3}). This contradicts the fact that $L$ is the true leader of the configuration. \qed
\begin{figure}[ht!] 
     \centering      
     \includegraphics[width=0.26\linewidth]{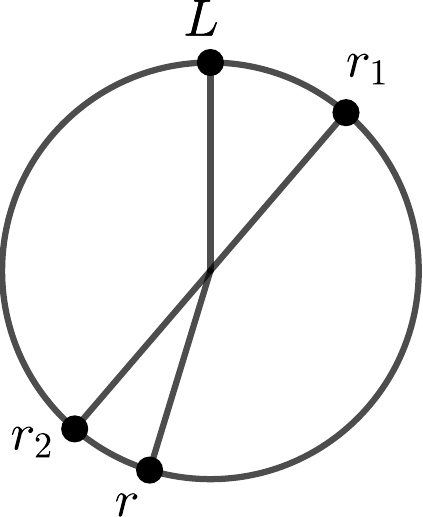}
     \caption{An image related to Proposition~\ref{lemma3oo}}
     \label{Fig:3}
    \end{figure}
 \end{proof}

\section{The proposed algorithm}

Here we present an algorithm in Algorithm~\ref{algo:gathering} to gather the robots. Each robot has the palate consisting the colors:  \texttt{leader\_present}, \texttt{leader\_absent}, \texttt{verify}, \texttt{cognizant}, \texttt{undecided}, \texttt{off}. Initially, all colors of the robots' are set at \texttt{off}. If a robot cannot see any robot then it moves $\pi/2$ distance clockwise. If there is no multiplicity point in its visibility, then the robot verifies whether it is a cognizant leader or not. If it is a cognizant leader then it changes its color to \texttt{cognizant} and move to its clockwise neighbor. If the robot is an undecided leader then it verifies that its clockwise neighbor is safe to move or not. Before for further we give the definition of a \textit{safe} clockwise neighbor.

\begin{definition}[Safe neighbor]
    Suppose $r$ is an undecided leader and $s$ is the first clockwise neighbor of $r$. The robot $s$ is said to be a safe neighbor of $r$ if the first clockwise neighbor of the true leader of $\mathcal C_1(r)$ configuration is not antipodal to $s$.
\end{definition}

\begin{algorithm}[ht!]
 %\DontPrintSemicolon % Some LaTeX compilers require you to use \dontprintsemicolon instead
  \footnotesize

% \KwIn{The set of points occupied by robots visible to $r$} 
%  \KwOut{Destination point for robot $r$}
 \eIf{there is a robot visible}
 {
    \uIf{there is no multiplicity point}
    {
        \uIf{the robot $r$ is a cognizant leader}
        {
            Change the color to \texttt{cognizant}\;
            Move to the clockwise neighbor\;
        }
         \uElseIf{the robot $r$ is an undecided leader}  
         {
            % Change the color to \texttt{undecided}\;
            \uIf{the clockwise neighbor of $r$ is safe}
            {
                Change the color to \texttt{undecided}\;
                Move to the clockwise neighbor\;
            }
            \ElseIf{$\mathcal C_0(r)$ configuration does not have another undecided leader other than $r$}
            {
                \uIf{the clockwise neighbor has color \texttt{off}}
                {
                    Change the color to \texttt{verify}\;
                }
                \uElseIf{the clockwise neighbor of $r$ has color \texttt{leader\_absent}}
                {
                    Change the color to \texttt{cognizant}\;
                    Move to its clockwise neighbor\;
                }
                \ElseIf{the clockwise neighbor of $r$ has color \texttt{leader\_present}}
                {
                    Change the color to \texttt{off}\;
                }
            }
         }
         \uElseIf{the robot $r$ is a follower robot} 
         {
            \If{the counter-clockwise neighbor has color \texttt{verify}}
            {
                Let $r_0$ be the robot with color \texttt{verify}\;
                Let $s$ and $s_0$ be the antipodal positions of $r$ and $r_0$ respectively\;
                \eIf{there is a robot in $[s_0,s)$}
                {
                    Turn the light to \texttt{leader\_present}\;
                }
                {
                    Turn the light to 
                    \texttt{leader\_absent}\;
                }
            }
         }
    }
    \uElseIf{there is a multiplicity point but the robot $r$ is not at any multiplicity point}
    {
        \If{its clockwise or counter-clockwise neighbor is a multiplicity point}
        {
            Move to the closer multiplicity point\;
        }
    }
    \ElseIf{there is another visible position that is a multiplicity point}
    {
        \If{clockwise angular distance from the multiplicity point is $<\pi$}
        {
            Move to the multiplicity point\;
        }
    }
    
 }
 {
    Move $\pi/2$ distance in clockwise direction\;
 }
\caption{\footnotesize Gathering algorithm for $\pi$-visibility; executed by a generic robot $r$ with initial light color \texttt{off}\;}
\label{algo:gathering}
\end{algorithm}

If the clockwise neighbor of the undecided leader is a safe neighbor then the robot changes its light's color to \texttt{undecided} and moves to the clockwise neighbor. If the clockwise neighbor is not a safe neighbor then it verifies whether its $\mathcal C_0$ configuration consists another undecided leader or not. If its $\mathcal C_0$ does not consist another undecided leader, it signals the neighbor to verify whether its antipodal position is occupied or not. It changes its color to \texttt{verify}. 

Suppose a follower robot $r$ sees its counter-clockwise neighbor, say $r_0$ with color \texttt{verify}. Let $s$ and $s_0$ be the antipodal positions of the $r$ and $r_0$ respectively. The the follower robot verifies whether there is a robot in the interval $[s_0,s)$. If there is a robot then it changes its color to \texttt{leader\_present} otherwise, changes its color to \texttt{leader\_absent}. Now if a robot sees its clockwise neighbor has color \texttt{leader\_absent}, then it changes its color to \texttt{cognizant} and moves to its clockwise neighbor. If it sees its clockwise neighbor has color \texttt{leader\_present}, then it changes its color to \texttt{off}. 

Now, suppose there is a multiplicity point and a robot $r$ can see the multiplicity point. If $r$ is not at the multiplicity point and its clockwise or counterclockwise neighbor is a multiplicity point then the robot moves to the closer multiplicity point. Suppose $r$ is at a multiplicity point and there is another multiplicity point visible. In this case if the clockwise distance of the multiplicity point from $r$ is less than $\pi$, then $r$ moves to the another multiplicity point.

\section{Correctness of the proposed algorithm}
First, we categories a rotationally asymmetric configuration with no multiplicity points in the following:
 \begin{itemize}
 \item[\ding{118}] \textit{Configuration-A:} Only the expected leader is the true leader of the configuration. If the expected leader is an undecided leader then its clockwise first neighbor is safe.
 \item[\ding{118}] \textit{Configuration-B:} There are two expected leaders in the configuration.
 \begin{itemize}
     \item \textit{Configuration-BI:} One cognizant leader and one undecided leader.
    % when the undecided leader which is not the true leader finds its clockwise first neighbor safe.
     \item \textit{Configuration-BII:} Two undecided leaders.
     
     % when the undecided leader which is not the true leader finds its clockwise first neighbor unsafe.
 \end{itemize}
 \item[\ding{118}] \textit{Configuration-C:} One expected leader which is an undecided leader sees that its clockwise first neighbor is not safe.
 \end{itemize}

 \begin{lemma}\label{Nonesure}
  If the initial configuration is type configuration-A, then after finite-time execution of Algorithm~\ref{algo:gathering} at least one and at most two multiplicity points will form or, it turned into a configuration-C. If two multiplicity points are formed then they will be non antipodal to each other.
 \end{lemma}
 \begin{proof}
     In the initial configuration, let $L$ be the only expected leader which is either a cognizant leader or an undecided leader which finds its first clockwise neighbor safe. Then from the Algorithm~\ref{algo:gathering}, $L$ moves to its clockwise neighbor, say, $r$. All robots other than $L$ in the initial configuration are follower robots. If, on activation, a robot recognizes itself as a follower robot, then it does not move according to the algorithm. Therefore, until $L$ starts its move towards $r$, no other robot will move. While $L$ is moving towards $r$, the leading angle of $L$ becomes strictly less than $r$ and $r$ can see it. Also, there is no counterclockwise movement of robots unless a robot approaches towards a multiplicity point. Thus, while $L$ is moving towards $r$, the leading angle remains strictly larger than $r$. Thus, $r$ remains a follower robot and does not move unless it sees a multiplicity point somewhere else. Hence, a multiplicity point will form where $r$ is located in the initial configuration either if $L$ is not stopped by the adversary before it finishes its move or another multiplicity point is formed somewhere else. 

     Suppose $L$ is stopped by the adversary while moving towards $r$. If before the next time activation no other robots moves, then on next time activation $L$ remains the true leader. And, this time it becomes a cognizant leader even if it was not in the initial configuration. Thus, if no other robot decides to move meanwhile, then eventually $L$ will reach $r$ to form a multiplicity point.

     Now, let us investigate where else possibly a multiplicity point can be formed. While $L$ is moving, a robot $r_1$ initially in the antipodal position of $r$ can become an undecided leader because it cannot see the leading angle of $L$. Except for $r_1$, all robots can see the leading angle of $L$. If $r_1$ becomes an undecided leader then it will find its clockwise neighbor safe because in its $\mathcal C_1$ configuration $L$ is the true leader. So, it will start moving towards its clockwise neighbor say, $r_2$ (See Figure~\ref{fig:55}). Thus, either at $r$ or $r_2$ or, at both the locations multiplicity points will form. Also, locations of $r$ and $r_2$ are antipodal to each other.

      \begin{figure}
    \centering
    \includegraphics[width=0.28\linewidth]{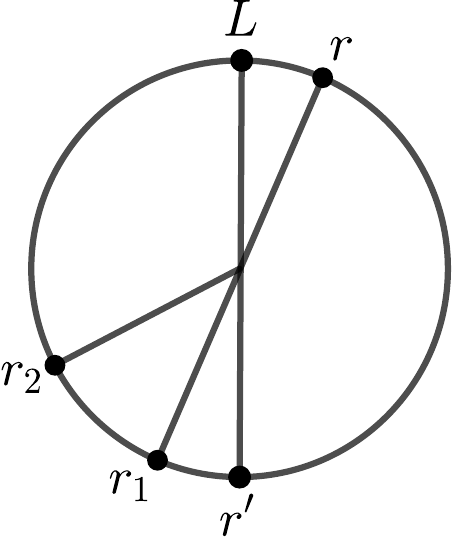}
    \caption{an image related to Lemma~\ref{twomult}}
    \label{fig:55}
    \end{figure}

     Since $L$ can see the leading angle of the $r_1$, so $L$ may lose its leadership if it is stopped while moving. Now, what if $L$ loses its leadership and $r_1$ is stopped before reaching $r_2$. In this case also, $r_1$ will become the true leader of the current configuration on next time activation. Also, every other robot can see the leading angle of the $L$ that is smaller than leading angle of them. So, no other robot will become an expected leader meanwhile.
     
     Now, if $r_1$ finds $r_2$ as a safe neighbor on next time activation then it will eventually move to $r_2$ to form a multiplicity point. Otherwise, the it will form a $configuration-C$. \qed

 \end{proof}

 \begin{lemma}\label{Nonecl}
 If the initial configuration is type configuration-C, then after finite time execution of Algorithm~\ref{algo:gathering} only one multiplicity point will form.
 \end{lemma}
 \begin{proof}
 When the initial configuration is of type configuration-C then the only expected leader, say $r$ is the undecided leader such that its first clockwise neighbor, say $r_1$, is not safe. In this case, according to the Algorithm~\ref{algo:gathering} the undecided leader changes its color to \texttt{verify}. After this, when first time the $r_1$ is activated, then on observing the color \texttt{verify} of its counterclockwise neighbor, it does the checking according to the Algorithm. In this configuration all the robots except $r$ are follower, so they do not move. So, on verification $r_1$ will must change its color to \texttt{leader\_absent}. Next time, when $r$ again gets activated, it observes the color \texttt{leader\_absent} of its clockwise neighbor. Then it changes its color to \texttt{cognizant} and move to its clockwise neighbor $r_1$. Even if $r$ gets stopped in the middle while moving, on next time activation $r$ will remain a true but an undecided leader which finds its clockwise neighbor unsafe. Thus, configuration type remains the same. But every time when $r$ gets activated it can see the \texttt{leader\_absent} of its clockwise neighbor. Thus, eventually $r$ will reach $r_1$. While $r$ is moving, except the antipodal robot of $r_1$, say $r_2$, all robots can see the leading angle of $r$, which is strictly smallest in the configuration. Also, $r_2$ will not become expected leader after $r$ starts its move. Because if a clockwise move of $r$ makes $r_2$ an expected leader then it was already an expected leader in the initial configuration, which is untrue.
 Hence, no other robot will become an expected leader and only one multiplicity point will form at the location of $r_1$. \qed
 \end{proof}

 \begin{lemma}\label{Noneslonecl}
 If the initial configuration is type configuration-B, then after finite time execution of Algorithm~\ref{algo:gathering} at least one and at most two multiplicity point will form. If two multiplicity points are formed then they will be non antipodal to each other.
 \end{lemma}
 \begin{proof}
      In configuration-B there are two expected leaders. Then there are two possibilities, either there is one cognizant leader and one undecided leader, or, both are undecided leaders. Let $L$ and $r$ be the cognizant leader and undecided leader of the initial configuration. Let the first clockwise neighbors of $L$ and $r$ be $r_1$ and $r_2$ respectively. 

      \ding{111} \textit{One cognizant leader and one undecided leader:} 

        \ding{118} \textit{Case-I} (When $r_2$ is not a safe clockwise neighbor of $r$) In this case, $L$ and $r$ must be antipodal. Also, $r_1$ and $r_2$ must be antipodal to each other (See Figure~\ref{fig:8}). According to Algorithm~\ref{algo:gathering}, $L$ on activation, decides to move to reach $r_1$. If $L$ starts moving before $r$ is activated, then leading angle of $L$ becomes strictly smallest. The $r$ can see $\lambda(L)$. Thus, every robot in the system can see $\lambda(L)$ except for $r_2$. But $r_2$ does not become an expected leader because it has lost leadership with $r$. Thus, if $L$ starts moving before $r$ is activated, no robot will move and a multiplicity will form at $r_1$. Now, if $L$ is stopped in the middle, then also on the next time activation $L$ remains the true leader of the configuration. Thus, a multiplicity forms at $r_1$.

         \begin{figure}[ht!]
     \centering
     \includegraphics[width=0.3\linewidth]{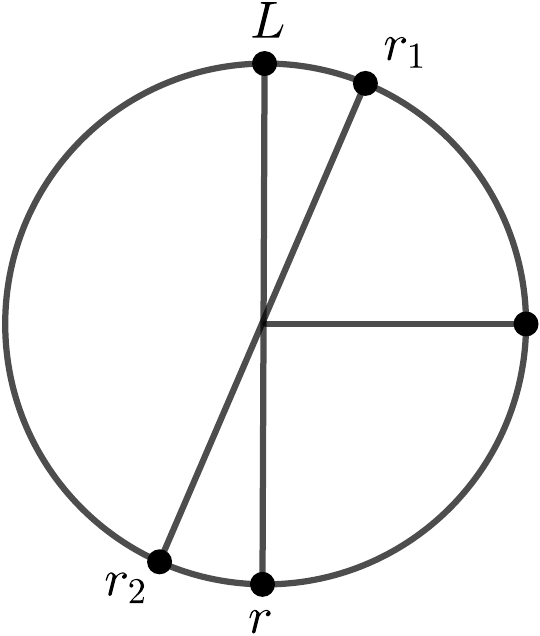}
     \caption{One cognizant leader $L$ and one undecided leader $r$ that finds its clockwise neighbor unsafe as $r_1$ and $r_2$ are antipodal to each other.}
     \label{fig:8}
 \end{figure}

    If $r$ wakes up before $L$ starts moving, it will change its color to \texttt{verify}. Upon seeing it $r_2$ will verify. Meanwhile, the $L$ might get activated. If $L$ wakes up then it will move to $r_1$ make a multiplicity point in the location of $r_1$. If $L$ did not reach $r_1$. Then on verification by $r_2$, it will see the true leader and turn on the light \texttt{leader\_present}. Then the robot $r$, on the next time activation, will turn its color \text{off}. And, only one multiplicity point will form at the position of $r_2$. If already $L$ reaches at $r_1$, then $r_2$ cannot see the leader. Thus, it will turn the color \texttt{leader\_absent}. But by that time a multiplicity point has been formed at the location $r_1$ which is visible by $r$. Also, before $L$ reaches $r_1$, no other robot will become an expected leader once $L$ starts to move. Thus, only one multiplicity point will form at the location of $r_1$. 

    \ding{118} \textit{Case-II} (When $r_2$ is a safe clockwise neighbor of $r$) In this case, if $L$ and $r$ get activated while the other has not started moving then both will move to their respective clockwise neighbors. Also, note that $r_1$ and $r_2$ are non antipodal to each other. First we show that no rotational symmetric configuration forms while the leaders are moving.

   If a rotationally symmetric configuration is created then robot $r$ has to move. Because if only $L$ moves then its angle sequence remains the strictly smallest one. Suppose that a rotationally symmetric configuration $\mathcal{C}_{sym}$ has been formed. If $L$ is moved to form $\mathcal{C}_{sym}$, then there will be another robot $L'\ne r$ other than $L$ that has same angle sequence, thus leading angle with $L$. Then it will constitute a contradiction because otherwise $L'$ would be the true leader in the initial configuration. Thus, $\mathcal{C}_{sym}$ can only be formed if adversary only moves $r$. Thus, in $\mathcal{C}_{sym}$ there is another robot $L'$ with same angle sequence as $L$. Let us consider two cases. First, let $\lambda(L)=\lambda(r)$ in the initial configuration. Then after any move the leading angle of $r$ becomes the strictly smallest. Thus, the new configuration must be rotationally asymmetric. Next, let $\lambda(L)<\lambda(r)$. In this case, $L$ and $L'$ both have strictly smaller leading angle than $r$. But, $r$ cannot miss the leading angle of the both $L$ and $L'$. Thus, $r$ would not be an expected leader in the initial configuration. Hence, adversary cannot make a rotationally symmetric configuration. 

    Let us note whether any other follower robot will become an expected leader. First suppose that $\lambda(L)=\lambda(r)$. Then $L$ and $r$ must be non antipodal. If one of the expected leader moves, then their leading angle is smaller than all the other follower robots. And the follower robots can see one of the expected leaders' leading angle. Thus, no expected leader will not become an expected leader. Next let $\lambda(L)<\lambda(r)$, then there are two cases: either $L$ and $r$ are antipodal or not antipodal. If $L$ and $r$ are not antipodal, then $r$ must be antipodal to $r_1$. In this case, all follower robot can see $\lambda(L)$. Thus, no follower robot will become an expected leader. If $L$ and $r$ are antipodal, then also all follower robot can see $\lambda(L)$. Thus, in all cases no follower robot will turn into an expected leader when the expected leaders are moving.

    Next, we show that no deadlock arises even if adversary stops an expected leaders in the middle of its move. First, we consider the case when the $\lambda(L)=\lambda(r)$. Then there are two possible cases. Either $r$ is antipodal to $r_1$ or not. If $r$ and $r_1$ are antipodal to each other, then $L$ can see the leading angle of the $r$ but $r$ cannot see the leading angle of $L$ (See Figure~\ref{fig:1}). So, if $r$ moves before $L$ is activated then $L$ will lose its leadership. And, $r$ becomes only expected leader. Thus, either Configuration of type $A$ or $C$ can occur if $r$ is stopped in the middle. From Lemma~\ref{Nonesure} and Lemma~\ref{Nonecl}, a multiplicity point will be formed at $r_2$. If $L$ starts to move and stopped in the middle then then it remains a cognizant leader unless it loses its leadership with $r$. Thus, either a multiplicity point will form at the location of $r_1$ or location of $r_2$ or at both the locations.

  \begin{figure}[ht!]
\centering
\begin{subfigure}{.5\textwidth}
  \centering
  \includegraphics[width=0.62\linewidth]{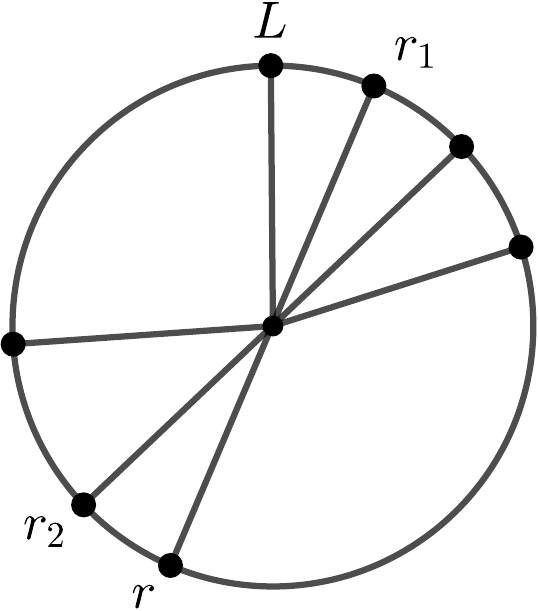}
     \caption{When $r$ and $r_1$ are antipodal}
     \label{fig:1}
\end{subfigure}%
\begin{subfigure}{.5\textwidth}
  \centering
  \includegraphics[width=0.65\linewidth]{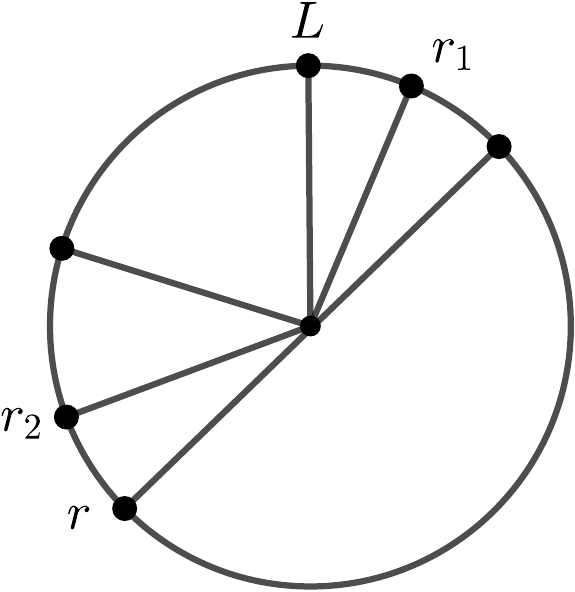}
     \caption{When $r$ and $r_1$ are not antipodal}
     \label{fig:2}
\end{subfigure}
\caption{A cognizant leader $L$ and an undecided leader $r$}
\label{fig:3}
\end{figure}
    
    If $r$ is not antipodal top $r_1$, then $L$ and $r$ can see each others leading angles (See Figure~\ref{fig:2}). Thus, depending on different scenario one of them can lose its leadership. In this case, if $L$ remains an expected leader then it will remain a cognizant leader. If $r$ remains an expected leader, then either it will become a true leader or an undecided leader. If the new configuration will be of type $A$ or $C$ then from Lemma~\ref{Nonesure} and Lemma~\ref{Nonecl}, a multiplicity point will be formed at $r_2$. If the new configuration is repeatedly of type $B$, then it will be with $\lambda(L)=\lambda(r)$ because $L$ and $r$ can see each others leading angle. In this case after a finite time a multiplicity will be formed either at the location of $r_1$ or $r_2$ or at both.

    Next, suppose that $\lambda(L)<\lambda(r)$. here either $L$ and $r$ are antipodal to each other or not (See Figure~\ref{fig:4}). If $L$ and $r$ are not antipodal then $r$ must be antipodal to $r_1$. In this if $L$ remains an expected leader then it will be a cognizant leader and a multiplicity point will form at the location of $r_1$. For the robot $r$ it may become the true leader or lose its leadership altogether, because after making a move it will be able to see the leading angle of the robot $L$. So, if after a move $r$ does not lose its leadership then either Configuration of type $A$ or type $C$ will occur. Thus, from Lemma from Lemma~\ref{Nonesure} and Lemma~\ref{Nonecl}, a multiplicity point will be formed at $r_2$.

\begin{figure}[ht!]
\centering
\begin{subfigure}{.5\textwidth}
  \centering
  \includegraphics[width=.65\linewidth]{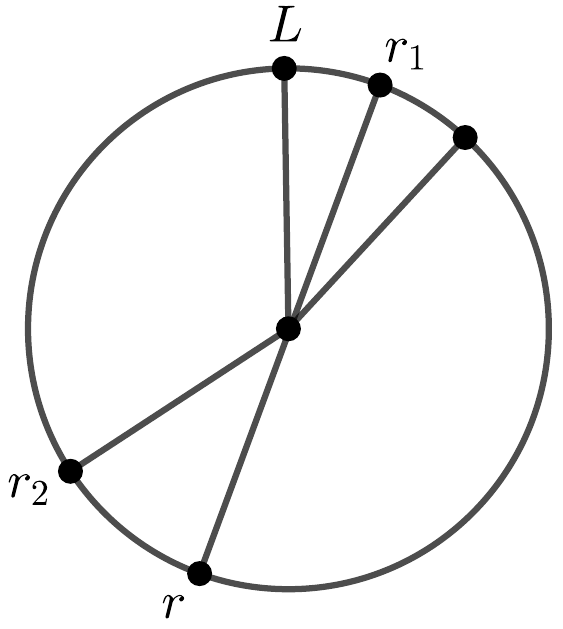}
  \caption{When $L$ and $r$ are not antipodal}
  \label{fig:4a}
\end{subfigure}%
\begin{subfigure}{.5\textwidth}
  \centering
  \includegraphics[width=.65\linewidth]{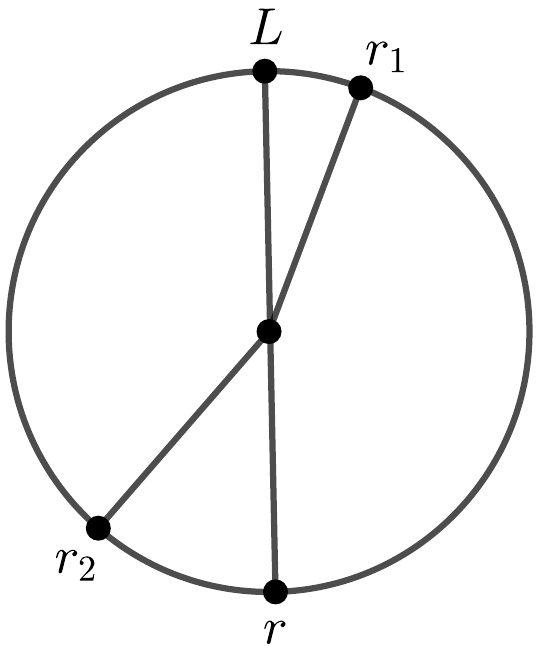}
  \caption{When $L$ and $r$ are antipodal}
  \label{fig:4b}
\end{subfigure}
\caption{Cognizant leader $L$ and undecided leader $r$}
\label{fig:4}
\end{figure}

    Now we consider the case when $\lambda(L)>\lambda(r)$ with $L$ and $r$ antipodal to each other. In this case, $L$ may lose its leadership or may become an undecided leader. If $L$ moves, the $r$ might lose its leadership. Otherwise, the robot $r$ may become a cognizant leader or an undecided leader. Thus, it may take the configuration of type $A$ or $C$  and can repeat the configuration type $B$ but this time both leaders will become undecided leaders. Configuration with two undecided leader is considered in the next case.

    \ding{111}\textit{Two undecided leaders:} From Proposition~\ref{lemma3oo}, $r_1$ and $r_2$ are non antipodal to each other. In this case, the true leader $L$ does not move. Because, for $L$ the condition at line~10 of the Algorithm. Becomes clockwise neighbor of $L$ is not safe and $\mathcal{C}_0(L)$, which is the current configuration, has an another undecided leader. For the another undecided leader $r$, it finds its clockwise neighbor safe. So, only $r$ will decide to move. Let us divide the further discussion into two cases: (i) $\lambda(L)=\lambda(r)$ and (ii) $\lambda(L)<\lambda(r)$. 

    For the first case (See Figure~\ref{fig:3}), as soon as $r$ starts moving, $r$ becomes the true leader. Also, since $L$ can see $\lambda(r)$, $L$ will lose its leadership. Next, only a follower robot at the antipodal position of $r_2$ cannot see the $\lambda(r)$. But will lose its leadership with $L$. Thus, no other follower robot will become an expected leader when $r$ starts moving. Next, suppose $r$ is stopped in the middle. Then new configuration will be either a configuration of type $A$ or of type $C$. From Lemma~\ref{Nonesure} and Lemma~\ref{Nonecl}, after a finite time a multiplicity will form at the location of $r_2$.

 \begin{figure}[ht!]
\centering
\begin{subfigure}{.5\textwidth}
  \centering
  \includegraphics[width=.65\linewidth]{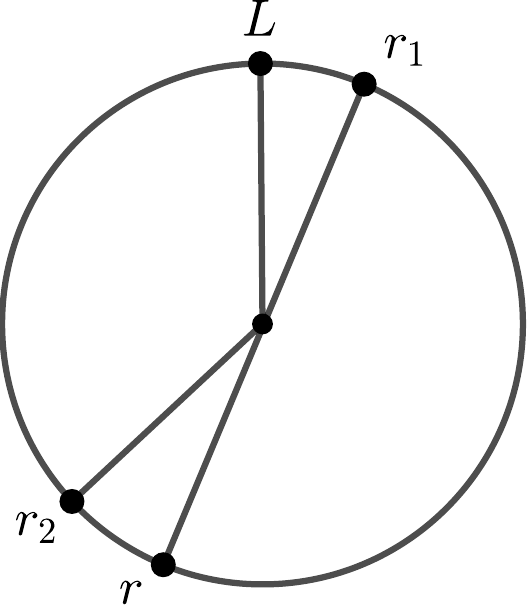}
  \caption{When $\lambda(L)=\lambda(r)$}
  \label{fig:3a}
\end{subfigure}%
\begin{subfigure}{.5\textwidth}
  \centering
  \includegraphics[width=.65\linewidth]{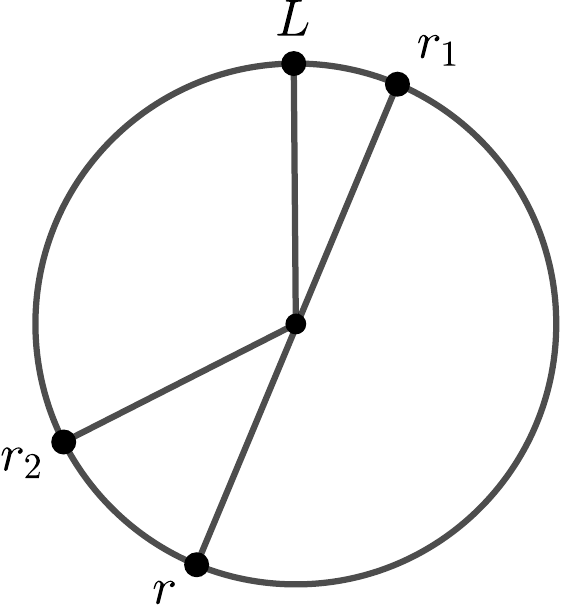}
  \caption{When $\lambda(L)<\lambda(r)$}
  \label{fig:3b}
\end{subfigure}
\caption{Two undecided leaders $L$ and $r$}
\label{fig:3}
\end{figure}

    For the second case, as soon as $r$ starts to move $L$ becomes a cognizant leader as long as $\lambda(L)<\lambda(r)$. If $\lambda(L)\ge\lambda(r)$, then since $L$ and $r$ can see each others leading angles. Thus, only one will remain an expected leader. Thus, new configuration will be either of type $A$ or $C$. From Lemma~\ref{Nonesure} and Lemma~\ref{Nonecl}, after a finite time a multiplicity will form at the location of $r_1$ or at the location of $r_2$, or at both the locations. \qed

    \end{proof}

 \begin{theorem}\label{twomult}
 From any rotationally asymmetric configuration with no multiplicity point, by finite time execution of Algorithm~\ref{algo:gathering} the robots can form at least one and at most two multiplicity points, and then all robots gather at a point on the circle.
 \end{theorem}
 \begin{proof}
Let $\mathcal{C}$ be a rotationally asymmetric configuration with no multiplicity point. Then there can be three exhaustive possible configurations, Configuration-A, Configuration-B and Configuration-C.
From lemma~\ref{Nonesure}, lemma~\ref{Nonecl}, and lemma~\ref{Noneslonecl}, we can say that in any type of configuration at least one and at most two multiplicity points will form. On formation of the multiplicity point, if a robot sees a multiplicity point, from line~27 of the algorithm robots one by one move to their closet multiplicity point. If only one multiplicity is formed, then all robots will eventually gather at there except one robot at the antipodal position of the multiplicity point. If that robot does not move at all, then eventually it will not see any other robot. Then according to line~33 of the algorithm, it will move $\pi/2$ in clockwise direction. After this it will be able to see the multiplicity point. If a robot is at a multiplicity point and can not see any other multiplicity point then it does not move anywhere further. If there are two multiplicity points created, then from Lemma~\ref{Nonesure} and Lemma~\ref{Noneslonecl}. the multiplicity points are non antipodal. Thus, each robot will be able to see at least one of the multiplicity point. Thus, eventually all robots will gather at those two multiplicity points. After this according to the line~29-31, robots of the one multiplicity point will move to another multiplicity point. If two robots simultaneously starts moving towards another multiplicity point, then a robot might see three multiplicity point in its view. But in such scenario the robot does nothing. Hence, after a finite time all robot will gather at a point. \qed
 
 \end{proof}

  Hence, we can conclude the following theorem.
  \begin{theorem}
 There exists a gathering algorithm that gathers any set of robots with finite communication and $\pi$ visibility from any initial rotationally asymmetric configuration under an asynchronous scheduler.   
 \end{theorem}

 \section{$\mathcal{FCOM}$ is not as powerful as $\mathcal{FSTA}$}\label{sec:comp}

In this section show that model $\mathcal{FCOM}$ is not as strong as $\mathcal{FSTA}$ under limited visibility. For sake of fairness, we keep the underlying topology and visibility model same as concerned in the proposed work. We define a problem \textsc{3to7} in the following definition. 

\begin{definition}
    Let there be two robots $r_1$ and $r_2$ placed on a circle such that $\alpha(r_1,r_2)=\pi/2$. The problem \textsc{3to7} asks only the robot $r_2$ to move clockwise to occupy the position such that $\alpha(r_1,r_2)$ becomes $7\pi/6$ (See Figure~\ref{fig:3to9}). The robots have $\pi$ visibility model and the robots movements are non-rigid. The robots agree on the clockwise direction.
\end{definition}

 \begin{figure}[ht!]
\centering
\begin{subfigure}{.5\textwidth}
  \centering
  \includegraphics[width=0.62\linewidth]{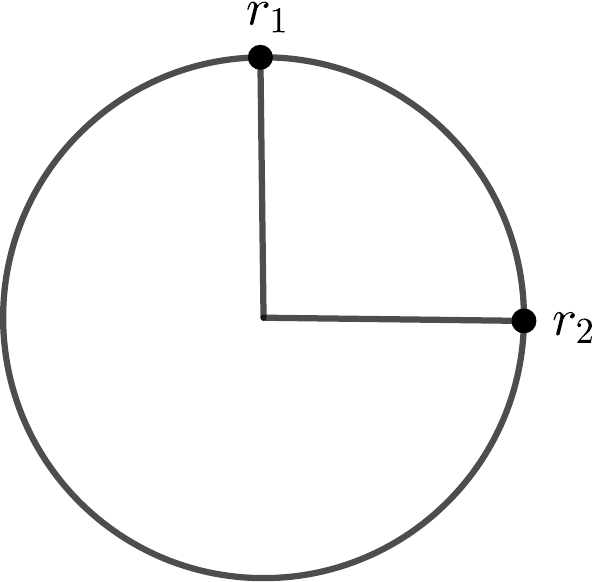}
     \caption{Initial configuration}
     \label{fig:31}
\end{subfigure}%
\begin{subfigure}{.5\textwidth}
  \centering
  \includegraphics[width=0.57\linewidth]{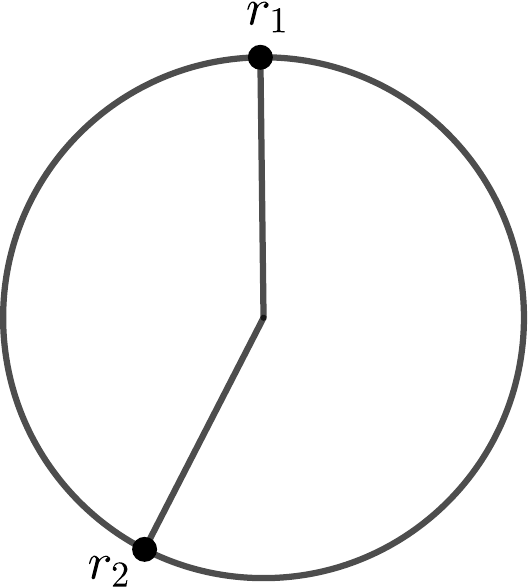}
     \caption{Final configuration}
     \label{fig:91}
\end{subfigure}
\caption{Configuration related to problem \textsc{3to7}.}
\label{fig:3to9}
\end{figure}

We show that this problem is solvable in $\mathcal{FSTA}$ under asynchronous scheduler but not solvable in $\mathcal{FCOM}$ under fully asynchronous scheduler. We propose the following Algorithm~\ref{algo:3to9} to solve the problem with two robots having two states \texttt{off} and \texttt{done}.

\begin{algorithm}
    \If{$state.r$ = \texttt{off}}
    {
        \eIf{ there is a robot $r'$ visible and $\alpha(r,r')=3\pi/2$}
        {
            $state.r\leftarrow$ \texttt{done}\;
            Move $2\pi/3$ angular distance clockwise\;
        }
        {
	Do nothing\;
        }
    }
    \If{$state.r$ = \texttt{done}}
    {
         \eIf{there is another robot $r'$ visible}
	{
        \If{$\alpha(r,r')\ne5\pi/6$}
	    {
            Move clockwise to make $\alpha(r,r')=5\pi/6$\;
        }
	}
	{
	    Move $\pi/6$ angular distance clockwise\;
	}
    }
    \caption{Executed by robot $r$; state of $r$ initially set to \texttt{off}}
    \label{algo:3to9}
\end{algorithm}

Next, we show that Algorithm~\ref{algo:3to9} solves Problem \textsc{3to7} under fully asynchronous scheduler.

\begin{lemma}\label{lm:xx}
    There is an algorithm that solves the problem \textsc{3to7} with two robots having model $\mathcal{FSTA}$ under fully asynchronous scheduler.
\end{lemma}
\begin{proof}
    Since the robots agree on the clockwise direction, if robots wake up and see the initial configuration, the robots can distinguish whether to move or not. On first time waking up, according to Algorithm~\ref{algo:3to9}, $r_2$ decides to move $2\pi/3$ angle clockwise. If it stopped in the middle then since its state is now \texttt{done}, either on next time activation $r_2$ can see $r_1$ (taken care in line~8) or will be at the antipodal position of $r_1$ (taken care in line~12). Eventually $r_2$ moves to attain $\alpha(r_1,r_2)=7\pi/6$.
    For the robot $r_1$, its state remains \texttt{off} because $r_2$ never moves at the point when $\alpha(r_1,r_2)=3\pi/2$. So, $r_1$ never moves as well.
\end{proof}

In the next Lemma~\ref{lm:x}, we show that the problem \textsc{3to7} is not solvable by robots with $\mathcal{FCOM}$ model even under fully synchronous scheduler.

\begin{lemma}\label{lm:x}
    There is no algorithm that solves the problem \textsc{3to7} with two robots having model $\mathcal{FCOM}$ even under fully synchronous scheduler.
\end{lemma}
\begin{proof}
    If possible, let there be an algorithm $\mathcal{A}$ that solves the problem under model $\mathcal{FCOM}$. To solve the problem the robot $r_2$ will have to cross the antipodal position of $r_1$. Since the robot movements are non-rigid, so the adversary can stop $r_2$ at the antipodal position of $r_1$. Since the lights of the robots are external, the view of both the robots becomes same. From this configuration if $r_2$ moves then $r_1$ will also move. This contradicts the correctness of $\mathcal A$.
\end{proof}

Hence, from Lemma~\ref{lm:xx} and Lemma~\ref{lm:x}, there is a problem that is solvable in $\mathcal{FSTA}$ but not solvable in $\mathcal{FCOM}$ under limited visibility. This, proves that $\mathcal{FCOM}$ is not as powerful as $\mathcal{FSTA}$.

 \section{Conclusion}
This work deals with the gathering of a set of robots with limited visibility deployed over a continuous circle. The robots are autonomous, anonymous, identical, and homogeneous. The robots operate through Look-Compute-Move cycle. Initially robots are are at distinct positions forming a rotationally asymmetric configuration. The robots are required to meet at a point not known beforehand. The robots have $\pi$ visibility model. Each robot cannot see the point at an angular distance $\pi$ from it. This work intends to extend the gathering problem under limited visibility model as full visibility is not practical due to hardware limitations. 

In this work, the robots agree on the clockwise direction. Each robot is equipped with a persistent external light ($\mathcal{FCOM}$) that can take colors from a predefined pallette consisting finitely many colors. The robots are operating under a fully asynchronous scheduler. The robots movements are non-rigid. This is the first work for $\pi$ visibility model under the fully asynchronous scheduler with non-rigid robots' movement. In addition, to validate our contribution, we attach a section showing that the model $\mathcal{FCOM}$ is not as powerful as model $\mathcal{FSTA}$.
As a future work, there are few questions listed below that are still open.
\begin{enumerate}
    \item Is it impossible to gather the robots with $\mathcal{OBLOT}$ in $\pi$ visibility under the fully-asynchronous scheduler?
    \item Is it possible to gather the robots with $\theta (<\pi)$ visibility model?
\end{enumerate}

\paragraph{\bf Acknowledgment}
The first would like to thank the University Grants Commission (UGC), the Government of India for support. The third author is supported by Science and Engineering Research Board (SERB), India.

\bibliographystyle{splncs04}
\bibliography{name}

\end{document}